\documentclass[letterpaper,twocolumn,10pt]{article}
\usepackage{usenix-2020-09}

\usepackage{tikz}
\usetikzlibrary{patterns}
\usepackage{amsmath}

\usepackage{filecontents}

\usepackage{algorithm}
\usepackage[noend]{algpseudocode}

\usepackage{pgfplots}
\usepackage{enumitem}
\usepackage{cleveref}
\usepackage{titlesec}

\newcommand{\ProtocolX}{Alea-BFT} 

\newcommand{\fillgap}[2]{$\langle \texttt{FILL-GAP}, #1, #2 \rangle$}
\newcommand{\filler}[1]{$\langle \texttt{FILLER}, #1 \rangle$}

\usepackage{pf2}
\newtheorem{theorem}{Theorem}
\newtheorem{lemma}{Lemma}
\newtheorem{definition}{Definition}

\usepackage[toc,page]{appendix}

\begin{document}

\date{}

\title{\Large \bf \ProtocolX: Practical Asynchronous Byzantine Fault Tolerance}

\author{
{\rm Afonso Oliveira}\\
IST (ULisboa) and INESC-ID
\and
{\rm Henrique Moniz}\\
Dune Analytics
\and
{\rm Rodrigo Rodrigues}\\
IST (ULisboa) and INESC-ID
} 

\maketitle

\begin{abstract}

Traditional Byzantine Fault Tolerance (BFT) state machine replication protocols assume a partial synchrony model, leading to a design where a leader replica drives the protocol and is replaced after a timeout. Recently, we witnessed a surge of asynchronous BFT protocols that use randomization to remove the assumptions of bounds on message delivery times, making them more resilient to adverse network conditions. However, these protocols still fall short of being practical across a broad range of scenarios due to their cubic communication costs, use of expensive primitives, and overall protocol complexity. In this paper, we present \ProtocolX, the first asynchronous BFT protocol to achieve quadratic communication complexity, allowing it to scale to large networks. \ProtocolX\ brings the key design insight from classical protocols of concentrating part of the work on a single designated replica, and incorporates this principle in a two stage pipelined design, with an efficient broadcast led by the designated replica followed by an inexpensive binary agreement. We evaluated our prototype implementation across 10 sites in 4 continents, and our results show significant scalability gains from the proposed design. 
\end{abstract}

\section{Introduction}\label{sec:intro}

The history of Byzantine fault tolerant (BFT) replication has gone through different stages throughout the years, from the initial exploration of the topic in the 1980s~\cite{Lamport}, then, in the lates 1990s, the start of a series of  practical protocols that achieve good performance~\cite{pbft}, and more recently the real-world adoption of this class of protocols in the context of Blockchains~\cite{hotstuff}.

Whenever a BFT protocol is devised, it needs to inevitably face the FLP impossibility result, which states that it is impossible to achieve consensus (or, equivalently, executing a replicated state machine command) in an asynchronous system with even a single fault~\cite{flp}. For many decades, the almost universally accepted way to circumvent this hurdle was by assuming a partial synchrony model, where the network is assumed to be initially asynchronous, but, after an unknown point in time, to deliver and process messages within a certain time bound~\cite{partialsynch}. This leads to a class of protocol designs where a leader can drive the execution of the protocol, but, after a timeout that indicates that the protocol is not making progress, all replicas cooperate in picking a new leader.

Recently, researchers picked up a different line of research that had been somewhat dormant for many years: asynchronous BFT protocols~\cite{aspnes2003randomized}. These protocols are safe and live irrespectively of any timing assumptions being met, but at the cost of those guarantees being probabilistic, i.e., they are provided with very high probability. Removing these timing assumptions brings the advantage of allowing the protocol to be more resilient against node and network delays, which may be due to reasons ranging from network problems to malicious activity.

The recently proposed algorithms in this model follow mostly from the work of HoneyBadgerBFT~\cite{honeybadger}, usually regarded as the first practical asynchronous BFT protocol and refined in subsequent work~\cite{dumbo}.
While these proposals were successful in showing that asynchronous BFT algorithms have reasonable performance and are resilient to adverse network conditions, they also fall short of being practical across a broad range of scenarios.
In particular, aspects such as their message complexity being cubic in the number of protocol replicas, the extensive use of  threshold encryption, and overall protocol complexity are at odds with deploying these proposals in systems with a large number of replicas.
Such scalability limitations prevent these asynchronous protocols from being actively deployed on emerging permissioned blockchain systems that require the system to scale to hundreds or even thousands of participants, for which all suitable solutions assume some level of synchrony~\cite{hotstuff, kauri}.

In this paper, we present \ProtocolX, the first protocol for asynchronous BFT state machine replication that is practical in real-world scenarios such as wide-area networks or deployments that go beyond a small number of replicas. The main insight in \ProtocolX\ is that it selectively brings a key design feature from classical partially synchronous protocols, namely the idea of having a per-command designated leader replica that drives the protocol execution for that command. By splitting the command execution in two phases, and placing on this replica the responsibility of the initial broadcast phase that disseminates each command, \ProtocolX\ can avoid redundant instances of expensive building blocks of asynchronous protocols, while simultaneously avoiding using threshold cryptography to encrypt proposals replicated across all replicas. However, this also introduces new challenges, namely that there is no guarantee that the broadcast by the leader will reach a sufficient number replicas in time for the subsequent agreement phase. We address this challenge by crafting an agreement phase where replicas either agree on the execution of the command in case the command was seen by enough replicas to reconstruct it if needed, or otherwise the command is locally stored in one of the various queues of pending commands that are eventually pushed through the agreement phase.

The resulting protocol is a significant leap forward in the state of the art of asynchronous BFT protocols, namely through the following characteristics.
\begin{itemize}
    \item It provides optimal resilience for the Byzantine model, tolerating up to $f=\lfloor\frac{N-1}{3}\rfloor$ faulty processes out of $N$ total processes;
    
    \item It is completely asynchronous, meaning that no assumptions are made regarding the delivery schedule of messages by the network, thus ensuring robustness under adversarial network conditions or attacks that break timing assumptions;

    \item It departs from a design where threshold cryptography is employed to encrypt proposals replicated across all the replicas, followed by an asynchronous common subset protocol and a final round of decrypting the results.
Instead, we propose a novel architecture with a two phase pipelined design, based on simple primitives, namely a single broadcast instance followed by a binary agreement.

    \item It provides significant asymptotic improvements over the state of the art protocols in this model. In particular, both the expected message and communication complexities can be reduced by a factor of up to $\mathcal{O}(N)$, while still terminating in constant expected time.
\end{itemize}

We implemented \ProtocolX\ and evaluated it in a wide area deployment, comprising replicas spread across 10 sites in 4 continents. Our results show that \ProtocolX\ brings significant improvements over the state-of-the-art protocols for asynchronous BFT, at the cost of a modest increase in latency due to leader rotation.

The remainder of the paper is organized as follows. \Cref{sec:related-work}  surveys related work. \Cref{sec:basics} describes the system model and key building blocks.
\Cref{sec:protocol} presents the design of \ProtocolX. \Cref{sec:efficiency} analyses its asymptotic complexity. \Cref{sec:impl} describes our implementation, which is evaluated in \Cref{sec:eval}. We conclude in \Cref{sec:conclusion}.

\section{Related Work}\label{sec:related-work}

The Byzantine consensus problem was formulated by Lamport et al.~\cite{Lamport}, and led to a series of proposals for Byzantine fault tolerant replication protocols~\cite{rampart,securering,BQS}. More recently, several proposals appeared that made BFT protocols more efficient, namely avoiding the use of expensive cryptographic signatures in the normal case~\cite{pbft,hq,zyzzyva}. BFT then gained a real-world adoption in the context of cryptocurrencies and blockchains, with several new protocols for that context~\cite{hotstuff}.

From these proposals, the subset that implement a form of consensus~--~namely state machine replication protocols~\cite{schneider1990implementing}~--~are faced with the FLP impossibility result~\cite{flp}, which states that there is no deterministic solution for the consensus problem in an asynchronous system, even with a single crash fault. To circumvent this result, almost all modern BFT systems, rely on timing assumptions such as partial synchrony~\cite{partialsynch} in order to ensure liveness. This is the case, for instance, of systems such as PBFT~\cite{pbft} and also more recent proposals such as HotStuff~\cite{hotstuff} or Kauri~\cite{kauri}.

As shown by Singh et al.~\cite{singh2008bft} through simulations where replicas are correct but the network is unreliable,
protocols in this model are sensitive to variations in message delivery performance, with the throughput dropping to zero under certain conditions, thus highlighting the need for protocols to take the network behavior into account.

As an alternative to protocols for the partially synchronous model, randomized protocols circumvent FLP by only guaranteeing the liveness property with high probability. The design for this class of protocols runs the main algorithm through multiple rounds until its non deterministic nature allows the probability of not having liveness to be irrelevant.
These protocols are then able to operate over a fully asynchronous model, therefore eliminating the need for timing assumptions and the consequences in terms of the fragility of the protocols in the presence of network unreliability.
Despite that fact that these properties make the protocols very interesting from a theoretical standpoint, asynchronous BFT protocols~\cite{rabin1983randomized,ben1994asynchronous,cachin2001secure,moniz2008ritas} have usually been considered impractical due to their high communication costs and expected termination time.
Very recently, several new randomized protocols appeared. At the core of this new line of proposals is the use of an asynchronous binary agreement (ABA) primitive, in which processes decide on the value of a single bit.
These ABA protocols are then used as building blocks for a solution to atomic broadcast and state machine replication. We next describe these new protocols, which form the most closely related work.

HoneyBadgerBFT~\cite{honeybadger} is based on the observation that atomic broadcast can be built on top of an asynchronous common subset (ACS) framework by combining it with a threshold encryption scheme. In ACS every party proposes an input value, and outputs a common vector containing the inputs of at least $N-f$ distinct parties.
As illustrated in Figure~\ref{fig:hb-acs}, HoneyBadgerBFT constructs ACS from the composition of two phases, namely reliable broadcast (RBC) and asynchronous binary agreement (ABA).
\begin{figure}[t]
\begin{center}
\includegraphics[width=\linewidth]{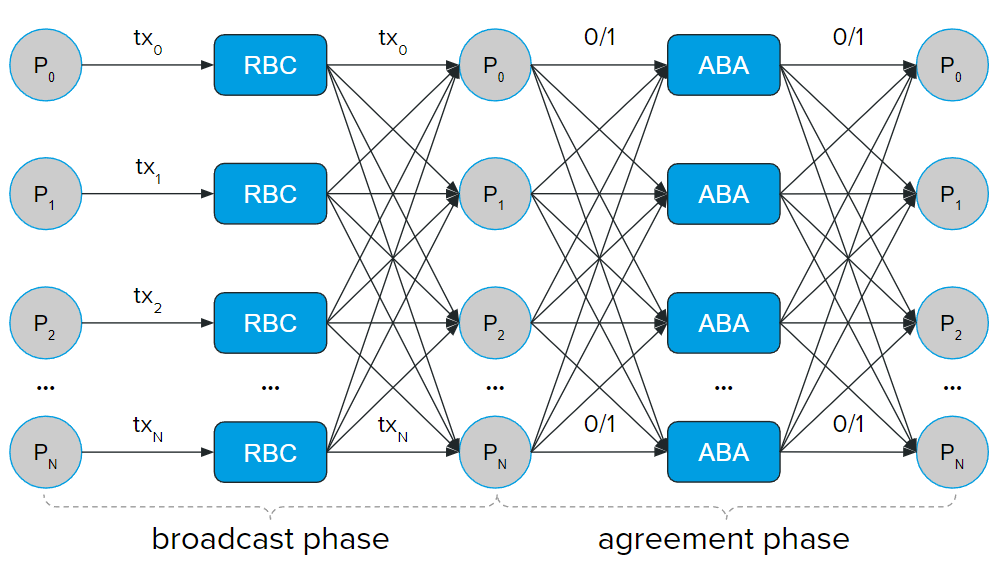}
\end{center}
\caption{The structure of ACS subprotocol used by HoneyBadgerBFT~\cite{honeybadger}. Each of the blue boxes corresponds to an instance of a distributed protocol, whose message exchanges are not depicted.}
\label{fig:hb-acs}
\end{figure}
During the broadcast phase, every replica starts an RBC instance in order to disseminate its proposal to all other replicas. Then, in the agreement phase, $N$ parallel ABA instances are invoked in order to decide on a $N$-bit vector, where the $i$-th value indicates whether or not to include the proposal from replica $P_i$ in the final ACS output. Here, the use of threshold encryption prevents an adversary from selectively censoring transactions, by selecting which proposals to include in the ACS output vector.

HoneyBadgerBFT significantly outperformed previous protocols, which constructed atomic broadcast directly from multi-valued Byzantine agreement (MVBA)~\cite{cachin2001secure}. This then led to the development of new protocols improving on it.
BEAT~\cite{beat} presents a series of techniques aimed at optimizing HoneyBadgerBFT and explores the trade-offs associated with using different broadcast and threshold cryptography sub-protocols.
EPIC~\cite{epic} is an extension to BEAT that provides security under the adaptive corruption model, meaning the adversary can adaptively decide whose replicas to corrupt at any moment during the execution of the protocol.
Finally, the Dumbo~\cite{dumbo} protocols present alternative ACS constructions aimed at reducing the number of ABA executions per command delivered. Dumbo1 reduces this number to a value $k$ independent of $N$ by selecting a committee of aggregators that nominate which subset of proposals to output, while Dumbo2 further reduces it to a constant by instantiating ACS based on a MVBA execution over reduced size inputs.

\section{Basics}\label{sec:basics}
In this section, we present the  system model and precisely define the basic blocks upon which our proposal is built.

\subsection{System model}
We consider a distributed system composed of $N$ processes, uniquely identified from the set $S = \{P_0, ..., P_{N-1}\}$, as well as an arbitrary number of clients.
We assume a Byzantine failure model where up to $f=\lfloor\frac{N-1}{3}\rfloor$ processes can fail during the execution of the protocol. The adversary is given full control over the behavior of a static set of these faulty processes meaning that they can deviate arbitrarily from the protocol specification and even collude among each other in order to subvert the properties of the protocol. The remaining processes that do not fail during protocol execution are referred to as correct.
The system is asynchronous, with the delivery schedule of messages being delegated under adversarial control, and without bounds on communication delays or processing times. We consider the processes to be fully connected by channels providing guarantees that messages are not modified in transit and are eventually delivered. In practice, this requires message retransmission and point to point authentication, but by considering this network model we can omit these from the protocol description.
Lastly, the adversary is assumed to be computationally bound and therefore unable to subvert the cryptographic primitives employed.

\subsection{Specification}

We developed \ProtocolX\ as an atomic broadcast protocol, which is a commonly used abstraction for implementing state machine replication. Intuitively, this allows a process (e.g., a proxy replica) to broadcast a message (e.g., a state machine command) to all processes, ensuring that all processes deliver all messages in the same order (executing all commands in the same order and therefore transitioning through the same sequence of states). Formally, atomic broadcast is defined as follows~\cite{Hadzilacos1994AMA}:
\begin{itemize}[leftmargin=*]
    \item \textbf{Validity:}
    If a correct process broadcasts a message $m$, then some correct process eventually delivers $m$.

    \item \textbf{Agreement:}
    If any correct process delivers a message $m$, then every correct process delivers $m$.

    \item \textbf{Integrity:}
    A message $m$ appears at most once in the delivery sequence of any correct process.

    \item \textbf{Total Order:}
    If two correct processes deliver two messages $m$ and $m'$, then both processes deliver $m$ and $m'$ in the same order.
\end{itemize}

\subsection{Building blocks}
Similarly to other asynchronous replication protocols, \ProtocolX\ is designed in a highly modular way, by reusing several sub-protocols to carry out certain tasks.
In this modular, layered architecture, upper level protocols can provide inputs and receive outputs from sub-protocols at the lower layers of the stack.
Next, we present the precise specification of the underlying primitives that are used as building blocks for the \ProtocolX\ protocol.

\subsubsection{Verifiable Consistent Broadcast Protocol}
Verifiable consistent broadcast (VCBC) is a protocol to deliver a payload message from a distinguished sender to all replicas. It can only provide guarantees that all correct replica processes deliver the broadcast value if the sender is correct; however, it always ensures that no two correct processes deliver conflicting messages.
Additionally it allows any party $P_i$, that has delivered the payload message $m$, to inform another party $P_j$ about the outcome of the broadcast execution, allowing it to deliver $m$ immediately and terminate the corresponding VCBC instance.
More formally, a VCBC protocol ensures the following properties~\cite{cachin2001secure}:
\begin{itemize}[leftmargin=*]
    \item \textbf{Validity:}
    If a correct sender broadcasts $m$, then all correct parties eventually deliver $m$.
    
    \item \textbf{Consistency:}
    If a correct party delivers $m$ and another party delivers $m'$, then $m = m'$.
    
    \item \textbf{Integrity:}
    Every correct party delivers at most one message. Additionally, if the sender if correct, then the message was previously broadcast by it.

    \item \textbf{Verifiability:}
    If a correct party delivers a message $m$, then it can produce a single protocol message $M$ that it may send to other parties such that any correct party that receives $M$ can safely deliver $m$.
    
    \item \textbf{Succinctness:}
    The size of the proof $\sigma$ carried by $M$ is independent of the length of $m$.
\end{itemize}
In \ProtocolX\ we use a VCBC implementation consisting of a slightly modified version of an echo broadcast protocol~\cite{cachin2001secure} using threshold signatures to generate and validate a proof $\sigma$ associated with $M$. This allows the message size to be kept constant, ensuring succinctness. The message complexity of the VCBC protocol used is $\mathcal{O}(N)$ and its communication complexity is $\mathcal{O}(N(|m| + \lambda))$, assuming the size of a threshold signature and share is at most $\lambda$ bits.

\subsubsection{Asynchronous Binary Agreement}
An asynchronous binary agreement (ABA) protocol allows correct processes to agree on the value of a single bit. Each process $P_i$ proposes a binary value $b_i \in \{0, 1\}$ and decides for a common value $b$ from the set of proposals by correct processes. Formally, a binary agreement protocol can be defined by the following properties:
\begin{itemize}[leftmargin=*]
    \item \textbf{Agreement:}
    If any correct process decides $b$ and another correct process decides $b'$, then $b=b'$.
 
    \item \textbf{Termination:}
    Every correct process eventually decides.

    \item \textbf{Validity:}
    If all correct processes propose $b$, then any correct process that decides must decide $b$.
\end{itemize}
Given the FLP impossibility result~\cite{flp}, there is no deterministic algorithm capable of satisfying all the previous proprieties in the asynchronous model of \ProtocolX. A solution to this is problem is resort to a randomized solution that guarantees termination in a probabilistic way. As a result, the termination property is replaced with the following:
\begin{itemize}[leftmargin=*]
    \item \textbf{Termination:}
    The probability that a correct process is undecided after $r$ rounds approaches zero as $r$ approaches infinity.
\end{itemize}
This way, even though the total number of rounds required to reach agreement is unbounded, the probability that the protocol does not terminate converges to zero.

We instantiate this primitive via the Cobalt ABA~\cite{macbrough2018cobalt} protocol, a modified version of the protocol by Mostefaoui et al.~\cite{mostefaoui2014signature} to include a fix for a liveness issue present in the original protocol.
The protocol relies on a common source of randomness, a “common coin”, realized from a threshold signature scheme by signing a unique bit string, corresponding to the name of the coin, and combining the signature shares to generate a random seed~\cite{cachin2005random}.
It provides optimal resilience, $\mathcal{O}(N^2)$ message complexity, $\mathcal{O}(\lambda N^2)$ communication complexity and terminates in $\mathcal{O}(1)$ expected time.
\section{\ProtocolX}\label{sec:protocol}

In this section, we present the design of \ProtocolX, a new protocol that substantially improves the performance and scalability of asynchronous BFT replication. In addition to designing and implementing \ProtocolX, we also conducted a full correctness proof, which we present in Appendix~\ref{appendix:correctness}.

\subsection{Motivation}\label{sec:protocol:motivation}
We start by presenting in a systematic way the characteristics of state of the art proposals~\cite{honeybadger,beat,epic,dumbo} that may represent scalability barriers:

\begin{itemize}[leftmargin=*]
    \item \textbf{Message and communication complexity:}
    All the previously cited proposals employ an instance of ACS, where every replica must propose a candidate value out of which a subset is selected to be included in the final output vector. This requires all replicas to execute an all to all communication phase to disseminate their proposals. By itself, this step incurs a cubic message and communication costs due to the $N$ executions of reliable broadcast (RBC).
    
    \item \textbf{Number of ABA instances:}
    Since the agreement protocol (ABA) is randomized, it may take several attempts based on different random values to terminate. Previous work~\cite{dumbo} showed that, despite the expected number of rounds per ABA being constant, the overhead of running multiple ABA instances presents one of the major bottlenecks for existing protocols. The same work shows that it is possible to reduce the number of ABA executions to a constant value at the cost of extra communication steps, but we argue that this value could be further reduced to a single execution.

    \item \textbf{Bandwidth usage:}
    The strategy of sequentially executing independent ACS instances implies that some of the broadcast instances were effectively "wasted" as their values were not included in the final output and therefore must be broadcast again in a subsequent ACS execution. This results in unnecessary bandwidth usage.

    \item \textbf{Byzantine performance faults:}
    The adversary has absolute control over which replicas proposals are included in the final output, only constrained by the validity property of ACS. This means that up to $f$ proposals that originate from Byzantine replicas can be invalid, resulting in serious performance degradation.

    \item \textbf{Threshold encryption:}
    A reduction from ACS to atomic broadcast requires threshold encryption for the entire set of proposals, to ensure fairness. This also adds an extra all to all communication step, for broadcasting decryption shares for the proposals included in the ACS output vector.
\end{itemize}

\subsection{Overview and Intuition}\label{sec:protocol:overview}
Despite these drawbacks, asynchronous protocols based on an ACS framework showed that there is potential for asynchronous protocols to become practical, if we overcome the above limitations. 
A key idea that we leverage in  \ProtocolX\ to achieve this is to have a single replica propose a value per consensus instance, similarly to what happens in leader based protocols in the partially synchronous model, while all others simply agree on whether to deliver it or not. By itself, this enables us to both remove an all to all communication phase and only have a single ABA execution per command (or batch of commands).

\bigskip
\noindent \textbf{Strawman Proposal.}
This insight leads us to a strawman proposal, consisting of adapting the ACS construction of HoneyBadgerBFT, but instead of having all replicas simultaneously propose candidate values, a single replica is selected as the proposer for each consensus round. The role of the proposer is to select from its  buffer of pending commands a value (or batch of values) to serve as a proposal and broadcast it to all replicas, using a broadcast primitive that ensures that all replicas receive the same value (or no value at all).
Correct replicas would then proceed to execute a \textit{single} ABA to determine whether to deliver the proposed value for that round (if enough replicas have received it) or not deliver anything.
Additionally, the proposer could be deterministically rotated upon every ABA execution, in order to address the scenario where the proposer is faulty without introducing a fail-over sub protocol, similarly to what happens in other BFT protocols for the partially synchronous model~\cite{veronese2009spin} that incorporate leader rotation into the normal operation such that it is constantly changing.

This strawman protocol, however, raises an immediate problem. In previous protocols based on an ACS framework, replicas are guaranteed to receive proposals from at least $N-f$ correct replicas, and therefore they can wait until this threshold is met before deciding which values to input for the subsequent agreement stage. In contrast, in our strawman protocol, only a single replica takes the role of the proposer at any given time, so there is no way to determine whether the current proposer is faulty or not, thus making it difficult to decide which value to input into the ABA without resorting to some sort of timeout, which contradicts the asynchronous model.

\bigskip
\noindent \textbf{Final design.}
The impossibility of waiting for a specific threshold to be met before deciding the value to input to the ABA stage, leads us to the insight of not waiting at all, and instead allowing undelivered proposals to exist, which are then carried across rounds. In other words, every time a particular replica is reelected as the proposer, the corresponding ABA execution will decide over its backlog of pending proposals instead of a single newly proposed value (or batch of values).
This way, replicas can submit their input to start the ABA for a new round as soon as they conclude the previous round, since even if the decision is 0 (i.e., not deliver any proposal in the round), the same proposal will be eventually revisited when the same replica becomes the leader and a larger threshold of replicas become aware of the proposal, guaranteeing a convergence for an ABA decision of 1 over time. The ABA execution also serves a synchronization mechanism between replicas, since no replica can progress to a round until it has participated and terminated all ABA instances for previous rounds.

In \ProtocolX, we leverage this idea to decompose the monolithic architecture of previous ACS-based protocols, in which a binary agreement instance actively waits for the corresponding broadcast to terminate, into a two stage pipeline, where the results of the first phase (broadcast component) are queued to be eventually processed, either by the current or by a subsequent execution of the second phase (agreement component), which is being executed in parallel.
\begin{figure}[t]
\begin{center}
\includegraphics[width=\linewidth]{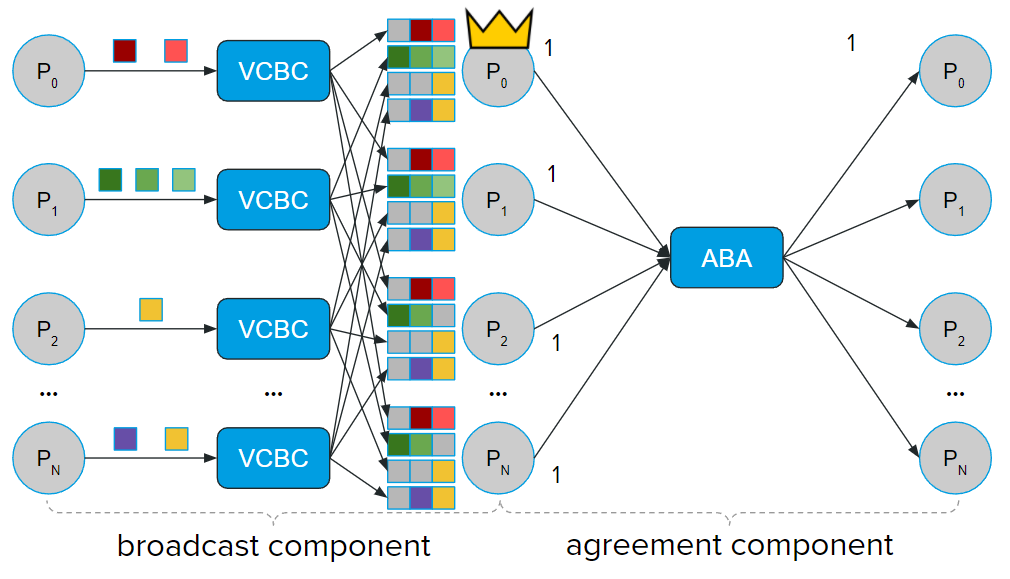}
\end{center}
\caption{Overview of \ProtocolX. Each request goes through a single instance of a broadcast primitive (VCBC), whose output is inserted in a priority queue structure at each replica, which determines the input to the final ABA protocol instance.}
\label{fig:alea}
\end{figure}

The resulting overall protocol flow is depicted in Figure~\ref{fig:alea}. It starts with the the broadcast component of the \ProtocolX\ pipeline, where replicas receive client commands, store these in a pending buffer of size $B$, and, when the buffer is full, disseminate its contents via a VCBC primitive tagged with an incremental sequence number $s$. The output of VCBC at each replica is stored in a buffer, and only removed upon a decision of 1 in the subsequent phase. 
This stage produces an instance of an ordered backlog of undelivered proposals at each replica. 
These instances are then used as input to the next component of the pipeline. Note that every replica maintains $N$ backlogs of undelivered proposals, one for each replica in the system, and these grow and shrink over time depending on how efficiently the agreement component can process them.

The next stage is the agreement component, which iteratively selects one of the backlogs and decides whether to deliver the oldest proposal contained in it.
To do so, replicas participate in a single ABA execution, voting $1$ if their backlog contains this proposal, or $0$ otherwise.
If the decision is $1$, this indicates that a sufficient threshold of correct replicas are aware of the proposal and may safely deliver it, as all other replicas are guaranteed to be able to actively fetch it if needed, via a recovery mechanism.
Otherwise, in case of a decision for $0$, the agreement component simply moves on to the next backlog, repeating the same process all over again.

As we mentioned, since the VCBC primitive of the broadcast component does not guarantee termination, we need to address the scenario where a correct process ouputs an ABA decision for 1, but does not yet know the corresponding proposal.
In this scenario, that correct process requests the missing proposal from the other processes that voted 1. This is guaranteed to work for the following reason: since ABA decided a value of 1, this implies that  at least one correct process voted for 1, and is therefore this process is able to compute a VCBC proof for the pending proposal and forward it to the requesting process.

\subsection{Description}
Now we explain the \ProtocolX\ protocol in greater detail.
Processes maintain two state variables shared between the two components of the pipeline: variable $S_i$, consisting of the set of all messages delivered by the protocol, which is initialized as empty upon a call to the \textproc{Start} procedure, and updated during the execution of the agreement component; and variable $\texttt{queues}_i$, comprising an array of $N$ priority queues, each corresponding to a distinct replica $P_x, \forall x \in \{0, ..., N-1\}$.
Algorithm~\ref{algo:alea:start} is responsible for initializing the shared state variables and starting the pipeline components upon a to call to the \textproc{Start} procedure. In the remainder of this section, we start by specifying the data structure of the priority queues, and then describe the two components of the pipeline in turn.

\begin{algorithm}[h]
\begin{algorithmic}[1]
\footnotesize

\algdef{SN}[constants]{Constants}{EndConstants}
{\textbf{constants:}}

\algdef{SN}[variables]{Variables}{EndVariables}
{\textbf{state variables:}}

\Statex
\Constants
    \State $N$ 
    \State $f$ 
\EndConstants

\Statex
\Variables
    \State $S_i \gets \emptyset$ 
    \label{algo:vars:s}
    \State $queues_i \gets \emptyset$ 
    \label{algo:alea:vars:queues}
\EndVariables

\Statex
\Procedure{START}{}
    \State $queues_i[x] \gets$ \textbf{new} pQueue() $,  \forall x \in \{0, ..., N-1\}$
    \State \textbf{async} \texttt{BC-START}() 
    \State \textbf{async} \texttt{AC-START}() 
\EndProcedure
\end{algorithmic}
\caption{\ProtocolX\ - Initialization (at $P_i$)}
\label{algo:alea:start}
\end{algorithm}

\subsubsection{Priority queues}
A priority queue is a custom data structure for storing elements, sorted according to their priority values.
We refer to each position in a priority queue as a slot, uniquely identified by a priority value associated with it, where the lower-numbered priority values represent the elements that must be processed first. Only a single element can ever be inserted in a given slot, even after being removed, as the slot is permanently labeled as used and cannot store another element.
There is a special slot called the head slot, that always points to the lowest priority slot whose value has not been removed yet. The pointer to the head slot progresses incrementally, conditioned by the insertion and removal of elements from the queue.
A priority queue exposes the following attributes:
\begin{itemize}[leftmargin=*]
    \item \textbf{id:}
    The unique identifier of the queue (static).
    
    \item \textbf{head:}
    The priority value associated with the head slot of the queue (dynamic).
\end{itemize}
Additionally, a priority queue provides an interface for accessing and modifying its contents as described below:
\begin{itemize}[leftmargin=*]
    \item \textbf{Enqueue $(v, s)$:}
    Add an element $v$ with a given priority value $s$ to the queue (ignored if the corresponding slot is not empty).
  
    \item \textbf{Dequeue $(v)$:}
    Removes all instances of the specified element $v$ from the queue, if it is present.
  
  
    \item \textbf{Peek $() \rightarrow \{v, \bot \}$:}
    Retrieve, without removing, the element $v$ in the head slot of the queue, or $\bot$ if the slot is empty.
\end{itemize}
As we will see, \ProtocolX\ leverages the proprieties of this structure to mediate the communication between the broadcast and agreement components of the protocol pipeline. In particular, each of the $N$ priority queues maintains the undelivered proposals coming from the other replicas, ordered by the priority value assigned to those proposals.

\subsubsection{Broadcast Component}
The broadcast component is responsible for establishing an initial local order over the client updates received and propagating that order to other replicas.
Every replica process maintains a two local state variables, a buffer of pending commands $buf_i$, and an integer value $priority_i$, indicating the next sequence number it should assigned to a proposal.
The main logic of this component, illustrated in Algorithm~\ref{algo:alea:bc}, is split between two upon rules:

\bigskip
\noindent \textbf{Upon rule 1 (\crefrange{algo:bc:upon:input}{algo:bc:upon:input:six}): }
The first rule is triggered, by any correct process $P_i$, upon the reception of a client message $m$ to be totally ordered by the protocol. It is responsible for selecting a batch of $B$ transactions from the pending buffer, attributing a local sequence number to it and broadcasting this pre-ordered proposal to all replicas. Process $P_i$ then proceeds as follows:
\begin{itemize}[leftmargin=*]
    \item
    If the set of delivered messages $S_i$ does not contain the client message $m$, append it to the buffer $buf_i$, or ignore it otherwise (\crefrange{algo:bc:upon:input:one}{algo:bc:upon:input:two}). 

    \item
    If the size of the buffer reached a threshold $B$, input $buf_i$ into a VCBC instance tagged with \textit{ID} $(i, priority_i)$ indicating that process $P_i$ assigned the local priority value $priority_i$ to a proposal consisting of the current buffer contents (\crefrange{algo:bc:upon:input:three}{algo:bc:upon:input:four}).
    
    \item Increment the value of $priority_i$, so that it can be assigned to the next proposal from $P_i$ and clear the buffer. (\crefrange{algo:bc:upon:input:five}{algo:bc:upon:input:six}).
\end{itemize}

\noindent \textbf{Upon rule 2 (\crefrange{algo:bc:upon:deliver}{algo:bc:upon:deliver:four}): }
The second rule is triggered, by any correct process $P_i$, upon the delivery of a proposal $m$ for a given VCBC instance tagged with \textit{ID} $(j, priority_j)$, where $j$ corresponds to the identifier of the replica $P_j$ that proposed $m$, and $priority_j$ to the sequence number assigned to it by $P_j$. Process $P_i$ proceeds as follows:
\begin{itemize}[leftmargin=*]
    \item
    Insert the delivered proposal $m$ into the slot $priority_j$ of the priority queue $Q_j$, mapping to $P_j$. This corresponds to adding the $m$ to the backlog of $P_j$ in $P_i$, with the priority value $priority_j$ (\crefrange{algo:bc:upon:deliver:one}{algo:bc:upon:deliver:two}).

    \item
    If the set $S_i$ contains $m$, indicating that it had already been delivered, then process $P_i$ immediately removes it from $Q_j$ in order to avoid ordering duplicate messages, effectively deleting $m$ and progressing through the backlog (\crefrange{algo:bc:upon:deliver:three}{algo:bc:upon:deliver:four}).
\end{itemize}

\noindent
Our protocol relies on batching to improve on performance and bandwidth utilization. In particular, in the implementation described in~\Cref{algo:alea:bc}, replicas aggregate a fixed number of commands, configurable by $B$, before assigning a priority value to the batch and forwarding to other replicas.

\begin{algorithm}[h]
\begin{algorithmic}[1]
\algdef{SN}[constants]{Constants}{EndConstants}
{\textbf{constants:}}

\algdef{SN}[variables]{Variables}{EndVariables}
{\textbf{state variables:}}

\algdef{SN}[upon]{Upon}{EndUpon}
    [1][]{\textbf{upon} #1 \textbf{do}}

\Statex
\Constants\label{bc:const}
    \State $B$
\EndConstants\label{bc:const-vars}

\Statex
\Variables\label{bc:vars}
    \State $buf_i$\label{algo:bc:vars:buff}
    \State $priority_i$ \label{algo:bc:vars:priority}
\EndVariables\label{bc:end-vars}

\Statex
\Procedure{BC-START}{}
    \State $buf_i \gets \emptyset$
    \State $priority_i \gets 0$
\EndProcedure

\Statex
\Upon[receiving a message $m$, from a client] \label{algo:bc:upon:input}
    \If{$m \notin S_i$} \label{algo:bc:upon:input:one}
        \State $buf_i \gets buf_i \cup \{m\}$ \label{algo:bc:upon:input:two}
        \If{$|buf_i| = B$} \label{algo:bc:upon:input:three}
            \State \textbf{input} $buf_i$ to \texttt{VCBC} $(i,  priority_i)$ \label{algo:bc:upon:input:four}
            \State $buf_i \gets \emptyset$ \label{algo:bc:upon:input:five}
            \State $priority_i \gets priority_i+1$ \label{algo:bc:upon:input:six}
        \EndIf
    \EndIf
\EndUpon\label{bc:end-update}

\Statex
\Upon[outputting $m$ for \texttt{VCBC} $(j,  priority_j)$]\label{algo:bc:upon:deliver}
    \State $Q_j \gets queues_i[j]$ \label{algo:bc:upon:deliver:one}
    \State $Q_j.Enqueue(priority_j, m)$ \label{algo:bc:upon:deliver:two}
    \If{$m \in S_i$} \label{algo:bc:upon:deliver:three}
        \State $Q_j.Dequeue(m)$ \label{algo:bc:upon:deliver:four}
    \EndIf
\EndUpon\label{bc:end-deliver}
\end{algorithmic}
\caption{\ProtocolX\ - Broadcast Component (at $P_i$)}
\label{algo:alea:bc}
\end{algorithm}

\subsubsection{Agreement Component}
The agreement component, presented in Algorithm~\ref{algo:alea:ac}, is responsible for establishing a total order among client commands. This is done through a succession of agreement rounds that iterate through the various priority queues, and decide whether to insert the head of that queue in the total order or skip it.
Processes maintain a single state variable $r_i$, serving as a unique identifier for the current agreement round. The execution of the agreement component starts with a call to the \textproc{AC-Start} procedure (line~\ref{algo:ac:start}), which initializes the local variable $r_i$ to $0$ and begins executing the agreement loop.



\bigskip
\noindent \textbf{Agreement loop (\crefrange{algo:ac:start:loop:start}{algo:ac:start:loop:end}): }
For each iteration $r_i$ of the agreement loop the backlog of proposals pertaining to a certain replica is selected.
This replica is a designated round leader, chosen through a deterministic function of the round number $F$ (e.g., by rotating through all replicas).
Let $P_a$ denote the current round leader, such that all replicas operate over the priority queue $Q_a$ for $r_i$.
A correct process $P_i$ proceeds as follows:
\begin{itemize}[leftmargin=*]
    
    \item
    Run an ABA instance tagged with ID $(r_i)$ to determine whether the $value$ in the head slot of $Q_a$ should be delivered for this round.
    Process $P_i$, inputs $1$ into the ABA if its local $Q_a$ contained a $value$ in the head slot, or $0$ otherwise (\crefrange{algo:alea:ac:start:loop:one}{algo:alea:ac:start:loop:four}).
    
    \item
    If the ABA execution decided for $0$, indicating that no proposal should be delivered for the current round $r_i$, simply proceed to the next loop iteration, otherwise:
    \begin{itemize}
        \item
        If process $P_i$ input $0$ into the ABA send a \texttt{FILL-GAP} message to all processes that voted for $1$. This is required because at this point in time $P_i$ is unaware of the value to deliver for $r_i$ and therefore must request it from another process
        (\crefrange{algo:ac:start:loop:fg:start}{algo:ac:start:loop:fg:end}).
        
        \item
        Block execution until the head slot of $Q_a$ contains a value to be delivered via a call to the \textproc{AC-Deliver} procedure. The value of the head slot can be updated by the delivery by a pending VCBC instance, either through "normal" execution or as result of the reception of a \texttt{FILLER} message (\cref{algo:ac:start:loop:block}).
    \end{itemize}
\end{itemize}

\noindent
In addition to the main agreement loop, the agreement component also defines two upon rules, associated with the recovery sub-protocol, to handle the reception of valid \texttt{FILL-GAP} and \texttt{FILLER} messages:

\bigskip
\noindent \textbf{Upon rule 1 (\crefrange{algo:ac:upon:fill}{algo:ac:upon:fill:five}):}
The first rule is triggered, by any correct process $P_i$, upon the reception of a valid \fillgap{q}{s} message from $P_j$, where $q$ identifies a priority queue $Q_q$, and $s$ specifies the current head slot of $Q_q$ in $P_j$. Process $P_i$ then proceeds as follows:
\begin{itemize}[leftmargin=*]
    \item
    Check if its local backlog pertaining to $P_q$ is more advanced than the one of $P_j$, by comparing the head pointer of its $Q_q$ against $s$ (\cref{algo:ac:upon:fill:three}). If it's lower this indicates that $P_i$ cannot satisfy the \texttt{FILL-GAP} request thus ignoring it, otherwise:
    
    \begin{itemize}
        \item
        Compute and store in $entries$ a verifiable message $M$, for all VCBC instances originating from $P_q$ tagged with a priority compromised between the value $s$ requested by $P_j$ and current head slot of $Q_q$ in $P_i$ (\cref{algo:ac:upon:fill:four}).
        
        \item
        Send a \texttt{FILLER} message to $P_j$ containing all the VCBC verifiable messages $M$, computed in the previous step (\cref{algo:ac:upon:fill:five}).
    \end{itemize}
\end{itemize}

\noindent \textbf{Upon rule 2 (\crefrange{algo:ac:upon:filler}{algo:ac:upon:filler:two}):}
The second rule is triggered, by any correct process $P_i$, upon the reception of a valid \filler{entries} message. This message is received as a response to a \texttt{FILL-GAP} request and contains the required information necessary for $P_i$ to progress in the execution of the protocol, by completing pending VCBC instances, after blocking in line~\ref{algo:ac:start:loop:block}. Process $P_i$ proceeds as follows:
\begin{itemize}[leftmargin=*]
    \item Deliver all $M$ messages in $entries$ to the corresponding VCBC instances.
    Note that the verifiability property of VCBC ensures that it immediately terminates upon the reception of $M$, therefore triggering the second upon rule of the broadcast component.
\end{itemize}

\noindent
Finally the \textproc{AC-Deliver} procedure (line~\ref{algo:ac:deliver}), called during the execution of the agreement loop, is responsible for delivering the contents of $values$, a batch of totally ordered messages $m$, to the application layer (line~\ref{algo:ac:deliver:app}).
Additionally, this procedure also removes $value$ from all priority queues and appends its contents to the set of delivered requests $S$.

\begin{algorithm}[h]
\begin{algorithmic}[1]
\footnotesize

\algdef{SN}[variables]{Variables}{EndVariables}
{\textbf{state variables:}}

\algdef{SN}[wait]{Wait}{EndWait}
    [1][]{\textbf{wait until} #1 \textbf{then}}
    
\algdef{SE}[DOWHILE]{Do}{doWhile}
    {\algorithmicdo}[1]{\algorithmicwhile\ #1}%

\algdef{SN}[for each]{ForEach}{EndFor}
    [1][]{\textbf{for each} #1 \textbf{do}}

\algdef{SN}[upon]{Upon}{EndUpon}
    [1][]{\textbf{upon} #1 \textbf{do}}

\Statex
\Variables\label{ac:vars}
    \State $r_i$ \label{ac:vars:r}
\EndVariables\label{ac:end-vars}

\Statex
\Procedure{AC-START}{}\label{algo:ac:start}
    \State $r_i \gets 0$
    \While {true}\label{algo:ac:start:loop:start}
        \State $Q \gets queues_i[F(r_i)]$\label{algo:alea:ac:start:loop:one}
        \State $value \gets Q.Peek()$\label{algo:alea:ac:start:loop:two}
        \State $proposal \gets value \neq \bot$ ? $1 : 0$ \label{algo:alea:ac:start:loop:three}
        \State \textbf{input} $proposal$ to \texttt{ABA} $(r_i)$ \label{algo:alea:ac:start:loop:four}
        \Wait[\texttt{ABA} $(r_i)$ delivers $b$]
            \If{$b = 1$}
                \If{$Q.Peek() = \bot$}\label{algo:ac:start:loop:}\label{algo:ac:start:loop:fg:start}
                    \State broadcast $\langle \texttt{FILL-GAP}, Q.id, Q.head \rangle$\label{algo:ac:start:loop:fg:end}
                \EndIf
                \Wait [$(value \gets Q.Peek()) \neq \bot$]\label{algo:ac:start:loop:block}
                    \State \textproc{ac-deliver}$(value)$
                \EndWait
            \EndIf
        \EndWait
    \State $r_i \gets r_i+1$
    \EndWhile\label{algo:ac:start:loop:end}
\EndProcedure\label{end-procedure:start}

\Statex
\Upon[receiving a valid $\langle \texttt{FILL-GAP}, q, s \rangle$ message from $P_j$] \label{algo:ac:upon:fill}
     \State $Q \gets queues_i[q]$\label{algo:ac:upon:fill:one}
     \If{$Q.head \geq s$}\label{algo:ac:upon:fill:three}
        \State $entries \gets $\texttt{VCBC}$(queue, s').M$ $\forall_{s'} \in [s, Q.head]$ \label{algo:ac:upon:fill:four}
        \State \textbf{send} $\langle \texttt{FILLER}, entries \rangle$ to $P_j$\label{algo:ac:upon:fill:five}
     \EndIf
\EndUpon

\Statex
\Upon[delivering a valid $\langle \texttt{FILLER}, entries \rangle$ message] \label{algo:ac:upon:filler}
    \ForEach [message $M \in entries$]\label{algo:ac:upon:filler:one}
        \State \textbf{deliver} $M$ to the corresponding \texttt{VCBC}\label{algo:ac:upon:filler:two}
    \EndFor
\EndUpon

\Statex
\Procedure{ac-deliver}{value}\label{algo:ac:deliver}
    \ForEach [$Q \in queues_i$]
        \State $Q.Dequeue(value)$
    \EndFor
    \ForEach [$m \in value$]
        \If{$m \notin S_i$}
            \State $S_i \gets S_i \cup \{m\}$
            \State \textbf{output} $m$ \label{algo:ac:deliver:app}
        \EndIf
    \EndFor
\EndProcedure
\end{algorithmic}
\caption{\ProtocolX\ - Agreement Component (at $P_i$)}
\label{algo:alea:ac}
\end{algorithm}

\section{Analysis}\label{sec:efficiency}

In this section we analyse the asymptotic efficiency of the \ProtocolX\ protocol, according to time, message and communication complexity metrics. The results of this analysis are summarized in Table~\ref{table:alea-complexity}.

To analyze \ProtocolX\ we observe that, for every particular proposal payload to be delivered, message exchanges occur in three different places.
First, during the execution of the broadcast component, a replica initiates a VCBC instance to disseminate the locally ordered proposal to all replicas, which is then queued in a priority queue slot according to the priority value assigned to it.
Second, all replicas participate in successive ABA executions, to decide whether or not to deliver the proposal in a particular slot. Here, we denote by $\sigma$ the average number of ABA instances executed over a single slot before it decides for $1$ and its contents are scheduled for delivery.
Finally, a fallback sub-protocol is triggered by replicas that did not VCBC-deliver the proposal before the corresponding ABA execution decided for $1$, in order to actively fetch it from the other replicas.

\subsection{Time Complexity}
Time complexity is defined as the expected number of communication rounds before a protocol terminates, or, when we consider a continuously running protocol such as atomic broadcast, from a command proposal to its output. In the case of \ProtocolX, the first and third steps terminate in constant time $\mathcal{O}(1)$, whereas the total number of rounds required for the agreement component to decide depend on the value of $\sigma$, therefore bounding the overall time complexity of \ProtocolX\ as $\mathcal{O}(\sigma)$.

\subsection{Message Complexity}
Message complexity is the expected number of messages generated by correct replicas during the execution of the protocol.
In \ProtocolX, the VCBC instance from the broadcast phase generates $\mathcal{O}(N)$ messages; then, every ABA instance exchanges $\mathcal{O}(N^2)$ messages; and finally the third recovery phase incurs an overhead of $\mathcal{O}(N)$ messages per replica that triggers this fallback protocol.
Hence, the message complexity of \ProtocolX\ is $\mathcal{O}(\sigma N^2)$, due to the $\sigma$ ABA instances that are executed per priority queue slot prior to delivery.

\subsection{Communication Complexity}
Communication complexity consists of the expected total bit-length of messages generated by correct replicas during the protocol execution.
Let $|m|$ correspond to the average proposal size and $\lambda$ the size of a threshold signature share.
The execution of VCBC incurs a communication complexity of $\mathcal{O}(N(|m|+\lambda))$, each ABA instance requires correct nodes to exchange $\mathcal{O}(\lambda N^2)$ bits, and finally each replica that triggers the recovery phase adds an additional communication cost of $\mathcal{O}(N(|m|+\lambda))$ bits.
This results in an expected total communication complexity of $\mathcal{O}(N^2(|m| + \sigma\lambda))$, due to the $\sigma$ ABA executions and up to $N$ recovery round triggers.

\begin{table}[htb]
\centering
\normalsize
\caption{Complexity of \ProtocolX\ decomposed by stages.}
\label{table:alea-complexity}
{
    \begin{tabular}{ | c | c | c | c | } \hline
    \textbf{Stage}  & Message                   & Communication                                       & Time                  \\ \hline
    Broadcast       & $\mathcal{O}(N)$          & $\mathcal{O}(N(|m|+\lambda))$                       & $\mathcal{O}(1)$      \\ \hline
    Agreement       & $\mathcal{O}(\sigma N^2)$ & $\mathcal{O}(\sigma\lambda N^2)$                    & $\mathcal{O}(\sigma)$ \\ \hline
    Recovery        & $\mathcal{O}(N^2)$        & $\mathcal{O}(N^2(|m|+\lambda))$                     & $\mathcal{O}(1)$      \\ \hline
    \textbf{Total}  & $\mathcal{O}(\sigma N^2)$ & $\mathcal{O}(N^2(|m| + \sigma\lambda))$             & $\mathcal{O}(\sigma)$ \\ \hline
    \end{tabular}
    }
\end{table}

\subsection{Estimating $\sigma$}
As previously mentioned, \ProtocolX\ does not guarantee a constant time execution, which ultimately could negatively affect the protocol latency. In particular, this is because multiple zero-deciding ABA instances could be executed over the same priority queue slot until its contents are considered to be totally ordered.
However, we argue that, despite being theoretically unbounded, the value of $\sigma$ (the number of ABA instances required for a decision) is in practice a very small constant, which ultimately is very close to the optimal value of 1. This statement is justified by the observation that, given a round-robin queue mapping function $F$, the same queue is revisited every $N$ epochs, meaning that $N-1$ other sequential ABA instances must have been executed by the time a particular queue is revisited.
Considering the validity property of ABA, which states that the decided value must have been proposed by a correct process, then the termination of a VCBC instance by $N-f$ correct replicas guarantees that the next ABA execution pertaining to it will decide for $1$.
Therefore, for the value of $\sigma$ to increase by a single unit, correct replicas would, on average, have to complete $N$ sequential ABA executions for every single VCBC instance, corresponding to the number of rounds after which a queue is revisited, an unlikely scenario given that ABA is a randomized protocol and consequently more expensive than a single deterministic execution of VCBC.
In Section~\ref{sec:eval}, we present practical results that further corroborate the hypothesis that the expected value of $\sigma$ is a small constant close to 1.

\section{Implementation}\label{sec:impl}
We implemented a prototype version of \ProtocolX, consisting of 5,000 lines of Java code.
The source code is organized in  a modular manner, with the main \ProtocolX\ outer protocol leveraging sub protocols (namely broadcast and binary agreement) as building blocks. Reliable point-to-point links were implemented using TCP streams, authenticated using HMAC with SHA-256 and 32 Byte keys shared by each replica pair.
Additionally, we implemented HoneyBadgerBFT and Dumbo1/2, which were used as comparison baselines for our experimental evaluation.

\section{Evaluation}\label{sec:eval}
In this section, we present the experimental evaluation of \ProtocolX.
We are interested in comparing the overall performance of \ProtocolX\ against state-of-the-art protocols in the same model for varying system sizes, loads and fault scenarios.
Additionally, we attempt to more precisely determine the analytical complexity of \ProtocolX\ by measuring the value taken by $\sigma$ under realistic deployment scenarios.

\subsection{Experimental setup}
For our experiments, we deployed \ProtocolX, Dumbo1/2 and HoneyBadgerBFT on $N=4$, $8$, $16$, $32$, $64$ and $128$ Amazon EC2 t2.medium instances, uniformly distributed across $10$ different regions (Paris, London, Frankfurt, Singapore, Tokyo, Mumbai, California, Virginia, Central Canada, and S\~{a}o Paulo) therefore spanning four continents. Each instance was equipped with $2$ virtual CPUs, $4$GB of memory and running Amazon Linux $2$.
We split our experiments in test groups based on the system scale $N$ and a varying batch size $B$ ranging from $4$ up to $10^6$ transactions. We used a fixed transaction size of $250$ bytes across all our experiments.
As baselines, we used HoneyBadgerBFT and the Dumbo family of protocols, thus capturing both the initial proposal in this new generation of asynchronous BFT protocols, and the state of the art proposals that were shown to outperform HoneyBadgerBFT.
We excluded the BEAT protocols from our experimental evaluation, as their work keeps the structure of HoneyBadgerBFT intact and most of the methods presented in it are orthogonal and compatible with the other ACS based protocols.

\subsection{Measuring $\sigma$}
In Section~\ref{sec:efficiency}, we presented a theoretical analysis of the complexity of \ProtocolX. The analysis showed that all complexity metrics are dependent on the value of a variable $\sigma$, corresponding to average the number of ABA executions per delivered proposal, which is theoretically unbounded.
In an attempt to quantify the actual value of $\sigma$ under realistic network conditions, we measured the number of messages generated by correct processes during protocol execution for different system sizes using a batch size of $1,000$ transactions, and in a scenario where all replicas constantly have requests to be executed, thus maximizing the load on the system.

Figure~\ref{fig:eval:sigma} compares the average number of messages generated by each correct process during the execution of \ProtocolX, in order to deliver a single proposal, against an analytical computation of this value for different $\sigma$ values.
As we can see, the experimental measurements follow very closely the computed values for $\sigma = 1$, independently of the system scale.
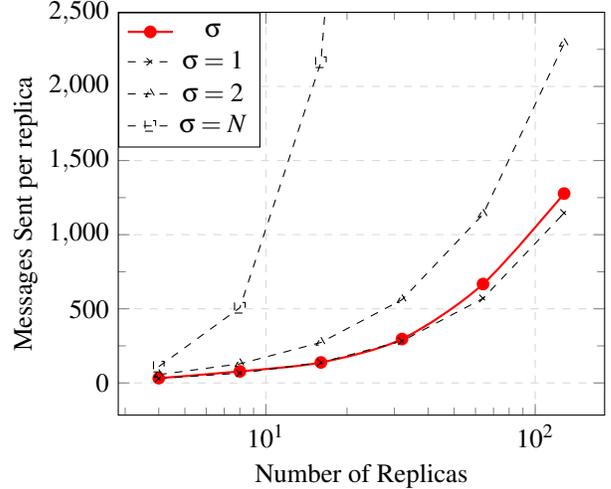
\begin{figure}[t]
\begin{center}
\begin{tikzpicture}
            \begin{axis}[
                width=.95\linewidth,
                grid = major,
                grid style = {dashed, gray!30},
                xmin = 0,
                ymax=2500,
                xmode=log,
                xlabel = {Number of Replicas},
                ylabel = {Messages Sent per replica},
                minor y tick num = 1,
                legend style={at={(0,1)}, anchor=north west}]
                
                \addplot
                [color=red,mark=*,smooth,thick] table[x=n,y=alea,col sep=comma] {Data/message-complexity.csv};
                \addlegendentry{$\sigma$}
                
                \addplot
                [color=black,mark=x,dashed] table[x=n,y=alea-s,col sep=comma] {Data/message-complexity-sim.csv};
                \addlegendentry{$\sigma=1$}
                
                \addplot
                [color=black,mark=triangle,dashed] table[x=n,y=alea-s2,col sep=comma] {Data/message-complexity-sim.csv};
                \addlegendentry{$\sigma=2$}
                
                \addplot
                [color=black,mark=square,dashed] table[x=n,y=alea-sn,col sep=comma] {Data/message-complexity-sim.csv};
                \addlegendentry{$\sigma=N$}
            \end{axis}
        \end{tikzpicture}
\end{center}\vspace{-2em}
\caption{Messages generated by \ProtocolX\ per replica and per batch delivered, for different values of $\sigma$.}
\label{fig:eval:sigma}
\end{figure}

A follow-up question is how does this translate, in practice, in terms of relative message complexity when compared to the state of the art. To this end, we measured experimentally the  message complexity per replica and per request batch of \ProtocolX\ in comparison with HoneyBadgerBFT and Dumbo1/2. 
As we can see in Figure~\ref{fig:eval:sigma-all}, this metric grows exponentially for the protocols based on ACS, but stays linear for \ProtocolX. This is expected since, in an ACS framework, every replica must reliably broadcast its proposals for that batch, which incurs $\mathcal{O}(N^2)$ messages per replica. In contrast, in \ProtocolX\ the broadcast primitive used has a message complexity of $\mathcal{O}(N)$.

\begin{table}[htb]
\centering
\normalsize
    \caption{Comparison of atomic broadcast protocols, assuming $\sigma$ is constant.}
    \label{table:alea-complexity-optimal}
{
    \begin{tabular}{ | c | c | c | c | } \hline
    \textbf{Protocol} & Message                   & Communication                                 & Time \\ \hline
    HBBFT             & $\mathcal{O}(N^3)$        & $\mathcal{O}(N^2|m| + \lambda N^3\log{}N)$    & $\mathcal{O}(\log{}N)$ \\ \hline
    Dumbo1            & $\mathcal{O}(N^3)$        & $\mathcal{O}(N^2|m| + \lambda N^3\log{}N)$    & $\mathcal{O}(\log{}k)$ \\ \hline
    Dumbo2            & $\mathcal{O}(N^3)$        & $\mathcal{O}(N^2|m| + \lambda N^3\log{}N)$    & $\mathcal{O}(1)$       \\ \hline
    \ProtocolX        & $\mathcal{O}(N^2)$        & $\mathcal{O}(N^2(|m|+\lambda))$               & $\mathcal{O}(1)$       \\ \hline
    \end{tabular}
    }
\end{table}

The previous experiments support our hypothesis that, despite being theoretically unbounded, under a wide-area deployment with substantial request load, the value of $\sigma$ does in fact approximate the optimal value of $1$.
In light of this, we present in Table~\ref{table:alea-complexity-optimal} the expected complexities of \ProtocolX, HoneyBadgerBFT and Dumbo1/2 when setting $\sigma=1$. As the table highlights, the gains are significant in terms of both message and communication complexity, allowing us to underline the importance of our proposal to the practicality of asynchronous BFT.

\subsection{Throughput}
Next, we measure the consequences of the previously studied protocol characteristics on their end to end performance, starting with the protocol throughput. Throughput is defined as the rate at which commands are serviced by the system, or, from the perspective of a Blockchain system, the number of transactions committed by unit of time.
In our experiments, we measured the throughput by launching multiple replicas executing the protocol at the maximum possible rate (i.e., always having pending client requests to process) and periodically registering the number of transactions committed during the last time interval.
Figure~\ref{fig:eval:batching}, compares the throughput of \ProtocolX, HoneyBadgerBFT and Dumbo1/2, for different system scales, as the batch size increases.
The results show that the positive slope of the \ProtocolX\ throughput curve is maintained for a wider range of batch sizes than with all the other protocols. This leads to a resulting peak throughput of \ProtocolX\ that is about one order of magnitude higher than the state of the art in asynchronous BFT.
This result is explained by the differences in communication complexity between the protocols, as presented in Table~\ref{table:alea-complexity-optimal}.
\begin{figure}[t]
\begin{center}
\begin{tikzpicture}
    \begin{axis}[
        width=.95\linewidth,
        grid = major,
        grid style = {dashed, gray!30},
        xmin = 0,
        xmode=log,
        ymode=log,
        xlabel = {Batch Size $(tx)$},
        ylabel = {Throughput $(tx/s)$},
        minor y tick num = 1,
        legend style={at={(axis cs:400,30000)}, anchor=north east}]
        
    \addplot
    [color=red,mark=*,smooth,thick] table[x=batch_size, y=avg, col sep=comma] {Data/Batching/alea-n4-tp.csv};
    \addlegendentry{\ProtocolX}
    
    \addplot
    [color=orange,mark=triangle,smooth,thick] table[x=batch_size, y=avg, col sep=comma] {Data/Batching/hb-n4-tp.csv};
    \addlegendentry{HBBFT}
    
    \addplot
    [color=green,mark=square,smooth,thick] table[x=batch_size, y=avg, col sep=comma] {Data/Batching/dumbo-n4-tp.csv};
    \addlegendentry{Dumbo1}
    
    \addplot
    [color=purple,mark=x,smooth,thick] table[x=batch_size, y=avg, col sep=comma] {Data/Batching/dumbo2-n4-tp.csv};
    \addlegendentry{Dumbo2}
    \end{axis}
\end{tikzpicture}
\end{center}\vspace{-2em}
\caption{Throughput with varying batch sizes.}
\label{fig:eval:batching}
\end{figure}
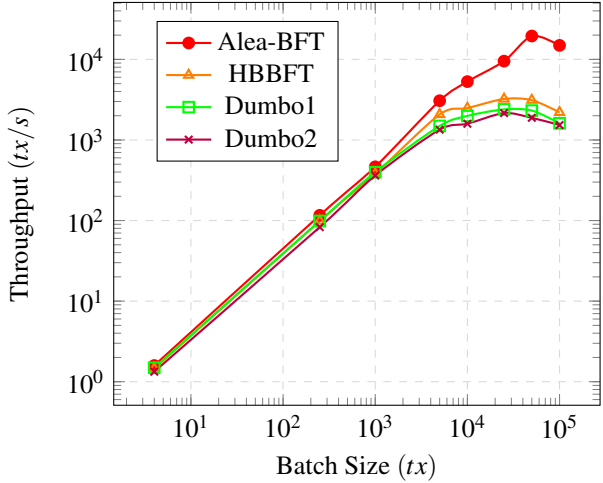

Next, we fix the batch size to a value of $5000$ transactions, and evaluate how the resulting throughput scales with the replica group size.
Note that this is a conservative experiment for \ProtocolX\ since our protocol would allow for further throughput gains than the other protocols by increasing the batch size, as previously shown.

\begin{figure}[h]
\begin{center}
\begin{tikzpicture}
    \begin{axis} [
        ybar,
        bar width=.1cm,
        width=0.95\linewidth,
        ymajorgrids = true,
        symbolic x coords={4,8,16,32,64},
        xtick=data,
        ymin=0,
        xlabel={Number of Replicas},
        ylabel={Throughput $(tx/s)$},
        legend style={
            at={(1,1)},
            anchor=north east,
            legend columns=1
        },
    ]
                
    \addplot table[x=n,y=alea,col sep=comma]{Data/tp-final.csv};
                
    \addplot [pattern=grid, pattern color=green] table[x=n,y=dumbo2,col sep=comma]{Data/tp-final.csv};
                
    \addplot [pattern=north east lines, pattern color=red] table[x=n,y=dumbo,col sep=comma]{Data/tp-final.csv};
                
    \addplot [pattern=dots, pattern color=black] table[x=n,y=hb,col sep=comma]{Data/tp-final.csv};
                
    \legend{\ProtocolX, Dumbo2, Dumbo, HBBFT}
    \end{axis}
\end{tikzpicture}
\end{center}\vspace{-2em}
\caption{Protocol throughput as the number of replicas varies.}
\label{fig:eval:throughput}
\end{figure}
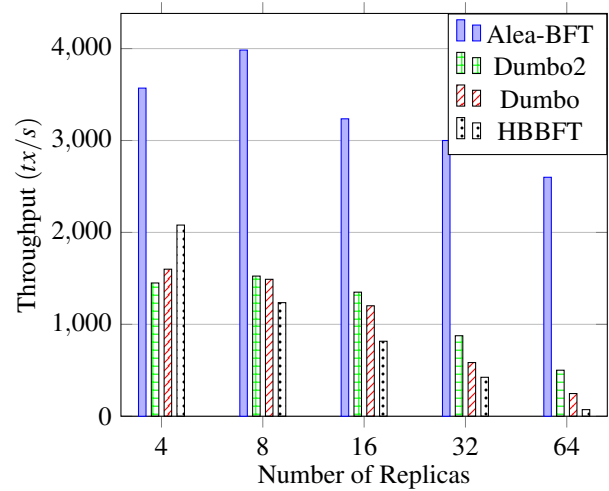

The results in Figure~\ref{fig:eval:throughput} show that \ProtocolX\ not only outperforms all other protocols for all system sizes, but this discrepancy in performance actually increases with the number of replicas in the system, thus showing better scalability. This is again expected since \ProtocolX\ presents lower message and communication complexities than its counterparts.
Another interesting remark is that for smaller values of $N$, HoneyBadgerBFT actually outperforms both Dumbo protocols despite being theoretically more expensive, showing that the benefit of reducing the number of ABA executions does not outweigh the overhead of the mechanisms required to achieve so for small system scales.

\subsection{Latency}
Next, we compare the latency of the various protocols, by measuring the time interval between the instant the first correct replica starts the protocol (i.e., selects the transaction from its pending buffer of client requests) until $(N-f)$ replicas deliver the result.
In HoneyBadgerBFT and Dumbo1/2 this corresponds to an instance of ACS plus the threshold decryption round, whereas in \ProtocolX\ it encompasses the full pipeline, including the period during which a transaction is waiting in the priority queues for the agreement component to select it for delivery.
\begin{figure}[t]
\begin{center}
\begin{tikzpicture}
    \begin{axis}[
        width=.95\linewidth,
        grid = major,
        grid style = {dashed, gray!30},
        xmin = 0,
        xlabel = {Number of Replicas},
        ylabel = {Basic Latency $(s)$},
        minor y tick num = 1,
        legend style={at={(axis cs:5,110)}, anchor=north west}
    ]

    \addplot
    [color=red,mark=*,smooth,thick] table[x=n,y=latency,col sep=comma] {Data/Latency/basic-lat-alea.csv};
    \addlegendentry{\ProtocolX}
                
    \addplot
    [color=orange,mark=triangle,smooth,thick] table[x=n,y=latency,col sep=comma] {Data/Latency/basic-lat-hb.csv};
    \addlegendentry{HBBFT}
                
    \addplot
    [color=green,mark=square,smooth,thick] table[x=n,y=latency,col sep=comma] {Data/Latency/basic-lat-dumbo.csv};
    \addlegendentry{Dumbo1}
                
    \addplot
    [color=purple,mark=x,smooth,thick] table[x=n,y=latency,col sep=comma] {Data/Latency/basic-lat-dumbo2.csv};
    \addlegendentry{Dumbo2}
    \end{axis}
\end{tikzpicture}
\end{center}\vspace{-2em}
\caption{Protocol latency as the number of replicas varies.}
\label{fig:eval:basic-latency}
\end{figure}
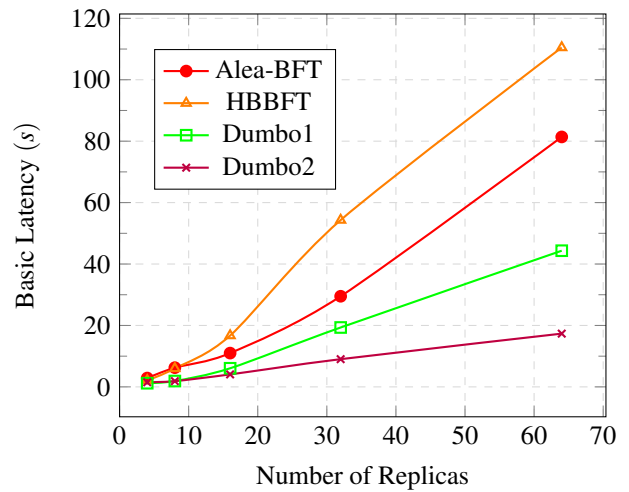
In Figure~\ref{fig:eval:basic-latency}, we examine the average latency of the protocols for different system sizes under no contention, i.e., having each node propose a single transaction at a time while no other proposals are being processed. As we can see, for small values of $N$, the basic latency of all four protocols is very similar. However, as the system size increases, we start to observe some differences, with the average latency of HoneyBadgerBFT increasing faster than the remaining protocols. This discrepancy can be explained by the latency overhead associated with running multiple ABA instances, a factor that is greatly reduced in both \ProtocolX\ and the Dumbo protocols. The average latency of \ProtocolX\ grows faster than that of Dumbo1 and Dumbo2.
This is because the latency of \ProtocolX\ is visibly influenced by the choice of replica that proposes the transaction, since the priority queues containing ordered proposals are traversed in a round robin manner, based on replica id. In particular, if the current leader is replica 0, then a proposal from a replica with a lower id will result in a lower latency measurement than if the proposal originated from a replica with a higher id. Note that, for fairness, the proposer was randomly selected and the results averaged across multiple runs.
\begin{figure}[t]
\begin{center}
\begin{tikzpicture}
    \begin{axis}[
        width=.95\linewidth,
        grid = major,
        grid style = {dashed, gray!30},
        xmin = 0,
        xmode=log,
        xlabel = {Throughput $(tx/s)$},
        ylabel = {Latency $(s)$},
        minor y tick num = 1,
        legend style={at={(0.05, 0.95)}, anchor=north west}
    ]
                
    \addplot
    [color=red,mark=*,smooth,thick] table[x=throughput,y=latency,col sep=comma] {Data/Ltp/latency-throughput-alea-32.csv};
    \addlegendentry{\ProtocolX}
                
    \addplot
    [color=orange,mark=triangle,smooth,thick] table[x=throughput,y=latency,col sep=comma] {Data/Ltp/latency-throughput-hb-32.csv};
    \addlegendentry{HBBFT}
                
    \addplot
    [color=green,mark=square,smooth,thick] table[x=throughput,y=latency,col sep=comma] {Data/Ltp/latency-throughput-dumbo-32.csv};
    \addlegendentry{Dumbo1}
                
    \addplot
    [color=purple,mark=x,smooth,thick] table[x=throughput,y=latency,col sep=comma] {Data/Ltp/latency-throughput-dumbo2-32.csv};
    \addlegendentry{Dumbo2}
    \end{axis}
\end{tikzpicture}
\end{center}\vspace{-2em}
\caption{Throughput vs.\ latency curves (n=32).}
\label{fig:eval:throughput-latency}
\end{figure}
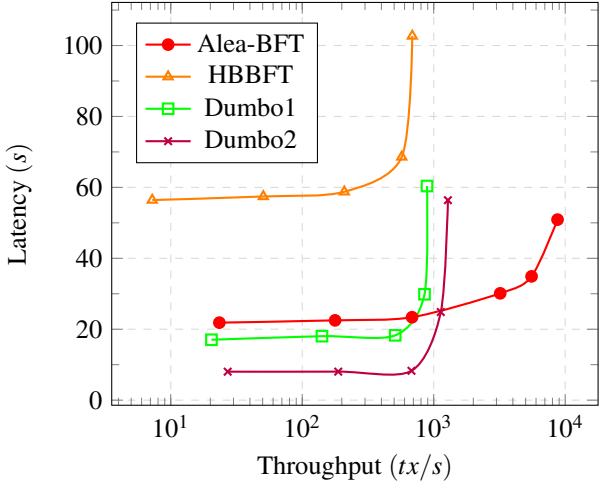
To gain an overall view of the tradeoff between latency and throughput for the various protocols, we show in Figure~\ref{fig:eval:throughput-latency} how the latency evolves as the system load increases for a medium system scale of $32$ replicas. Each point corresponds to measuring the throughput and latency of a given protocol for a fixed batch size, and the various points along a line correspond to varying the batch size and therefore increasing both latency and throughput as the batch size increases.
For all protocols, initially the latency stays relatively stable, only presenting small increases as the system load grows. However, as we reach the nominal capacity of each protocol, we see a very steep increase in latency as the system resources stop being able to keep up with the increase in system load. Note that \ProtocolX\ is able to sustain a stable latency for much higher system loads than all the other protocols, and this is due to its asymptotic complexity, which leads to a more optimized bandwidth usage.

\subsection{Performance under adversarial conditions}
Finally, we evaluate the the performance of \ProtocolX\ under adversarial conditions, particularly in the presence of a malicious network scheduler and faulty replica processes.

In the first experiment, we aim to understand how an adversary that controls the delivery of network messages can influence the value of $\sigma$, and consequently the performance of \ProtocolX. To this end, we conducted an experiment using a malicious scheduler that purposely delayed the delivery of VCBC instances by $N-f$ replicas, in order to artificially increase the value of $\sigma$ to a value that is closer to $N$. As illustrated in Figure~\ref{fig:eval:sigma-all}, even under adversarial network conditions, \ProtocolX\ still requires fewer message exchanges than HoneyBadgerBFT, although it requires more messages than both Dumbo protocols.
This is explained by the fact that, in a scenario where $\sigma = N$, both \ProtocolX\ and HoneyBadgerBFT require $N$ ABA executions per consensus instance, whereas Dumbo1 and 2 reduce this number to a small value $k$ (independent of  $N$) and a constant value, respectively.
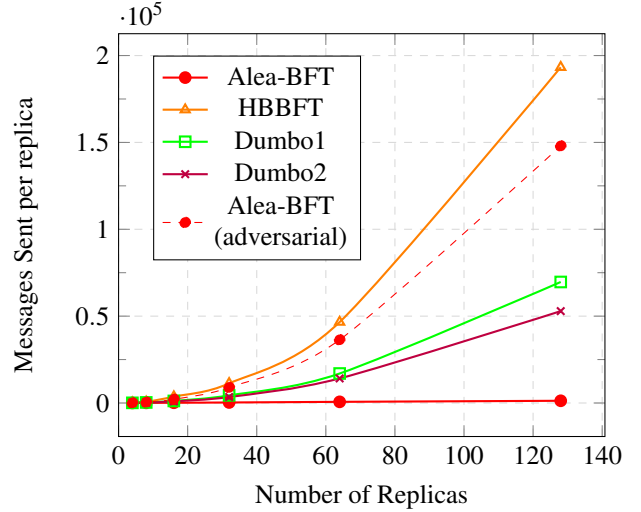
\begin{figure}[t]
\begin{center}
\begin{tikzpicture}
    \begin{axis}[
        width=.95\linewidth,
        grid = major,
        grid style = {dashed, gray!30},
        xmin = 0,
        xlabel = {Number of Replicas},
        ylabel = {Messages Sent per replica},
        minor y tick num = 1,
        legend style={at={(axis cs:10,200000)}, cells={align=center}, anchor=north west}
    ]
    
    \addplot
    [color=red,mark=*,smooth,thick] table[x=n,y=alea,col sep=comma] {Data/message-complexity.csv};
    \addlegendentry{\ProtocolX}
                
    \addplot
    [color=orange,mark=triangle,smooth,thick] table[x=n,y=hb,col sep=comma] {Data/message-complexity.csv};
    \addlegendentry{HBBFT}
                
    \addplot
    [color=green,mark=square,smooth,thick] table[x=n,y=dumbo,col sep=comma] {Data/message-complexity.csv};
    \addlegendentry{Dumbo1}
                
    \addplot
    [color=purple,mark=x,smooth,thick] table[x=n,y=dumbo2,col sep=comma] {Data/message-complexity.csv};
    \addlegendentry{Dumbo2}
    
    \addplot
    [color=red,mark=*,smooth,dashed] table[x=n,y=alea-n,col sep=comma] {Data/message-complexity.csv};
    \addlegendentry{\ProtocolX\ \\ (adversarial)}
    \end{axis}
\end{tikzpicture}
\end{center}\vspace{-2em}
\caption{Messages generated per replica per batch delivered under an adversarial network scheduler.}
\label{fig:eval:sigma-all}
\end{figure}

Additionally, we compare the performance of \ProtocolX\ and HoneyBadgerBFT in three different replica fault scenarios: failure-free, crash failure and Byzantine.
In the crash failure scenario, we configure $f$ replicas to completely ignore all external events.
In the Byzantine fault scenario, we simulate the attack described in \Cref{sec:protocol:motivation}, in which an attacker purposely submits invalid proposals to be ordered that despite consuming resources do not count for the overall throughput of the protocol.
For all experiments, the system scale was set to 4 replicas, with $f=1$, and the batch size to $1,000$ transactions.
\begin{figure}[t]
    \begin{center}
        \begin{tikzpicture}
            \begin{axis} [
                ybar,
                bar width=.5cm,
                width=.95\linewidth,
                ymajorgrids = true,
                legend style={at={(1,1)}, anchor=north east,legend columns=-1},
                symbolic x coords={Alea-BFT, HBBFT},
                xtick=data,
                ylabel = {Throughput $(tx/s)$},
                ymin=0,
                scaled y ticks = false,
                enlarge x limits=1
                ]
                
                \addplot table[x=protocol,y=free,col sep=comma]{Data/faults.csv};
                \addplot [pattern=grid, pattern color=red] table[x=protocol,y=crash,col sep=comma]{Data/faults.csv};
                \addplot [pattern=dots, pattern color=black] table[x=protocol,y=byzantine,col sep=comma]{Data/faults.csv};
                \legend{Fault-free, Crash, Byz.}
            \end{axis}
        \end{tikzpicture}
    \end{center}\vspace{-2em}
    \caption{Throughput of \ProtocolX\ and HoneyBadgerBFT for different fault scenarios.}
    \label{fig:eval:faults}
\end{figure}
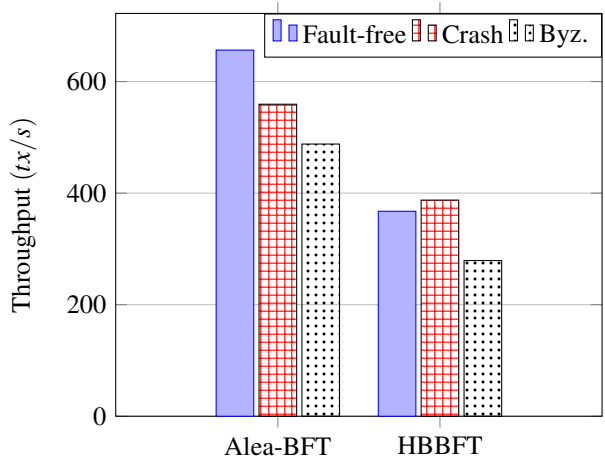
For the Byzantine experiments, both protocols show lower throughput values than in the failure-free scenario. For instance, the performance of HoneyBadgerBFT and \ProtocolX\  decrease by about 30\% and 24\%, respectively.
These results are directly related with how often the proposal from a Byzantine replica is delivered by the protocol.
In HoneyBadgerBFT, this corresponds to the threshold of proposals originating from Byzantine replicas that are included in the final output vector of ACS, which, under a fair network scheduler, follows a hypergeometric distribution over a population of size $N$ and $N-f$ draws. In contrast, under an adversarial network scheduler, the output of ACS will always contain $f$ Byzantine proposals, since it can be influenced by the delivery order of messages.
In \ProtocolX, the performance decay that occurs under the fault scenarios presented above follows directly from the queue selection function. In particular, the round robin strategy used in our implementation results in faulty replicas being selected periodically.
A possible improvement over our baseline implementation could explore the impact of using adaptive queue selection functions. Additionally, for the Byzantine fault scenario, it would be possible for \ProtocolX\ to ignore invalid proposals during the broadcast phase. In contrast, ACS based protocols cannot follow this approach since the commit of which proposals to include happens before the threshold decryption round revealing its contents.

\section{Conclusion}\label{sec:conclusion}

In this paper, we presented \ProtocolX, the first practical asynchronous BFT protocol to scale to large networks. \ProtocolX\ follows a principled design that splits the execution across two stages, and brings the idea of having a designated single replica drive the execution of the more expensive broadcast phase~--~to avoid redundant communication and computation that negatively affects the scalability of previous protocols~--~followed by a less expensive binary agreement phase. Our experimental evaluation shows that \ProtocolX\ is significantly more scalable than state of the art protocols, paving the way to the practical adoption of asynchronous BFT at scale.


\bibliographystyle{plain}
\bibliography{alea}

\begin{appendices}
\section{Correctness Proof}\label{appendix:correctness}
In this section, we present a correctness proof for the \ProtocolX\ protocol.
It is organized in four subsections, one for each property of the algorithm - validity, agreement, integrity and total order.
All proofs are written according to the structured proof format by Lamport~\cite{lamport_proof} and assume that $N=3f+1$.
For convenience, we start by presenting some auxiliary definitions, which will be in effect for all the proofs presented in this section.

\begin{definition}\label{definition:consensus-holds}
If a correct process $P_i$ enters round $r$, let $S_{r}^{(i)}$ denote the set of requests delivered by $P_i$ by the time it enters $r$. We say that consensus holds on entry to round $r$ if for any two correct processes $P_i$ and $P_j$ that enter $r$, then $S_{r}^{(i)} = S_{r}^{(j)}$. If consensus holds on entry to round $r$, and any correct process does enter round $r$, we denote by $S_r$ the common value of the $S_{r}^{(*)}$.
\end{definition}

\noindent
It is trivial to see that consensus always holds for round $0$, considering $S_0$ starts as empty. For this reason by proving the agreement and total order properties for a round where consensus holds, we can recursively generalise the results to every other round, using $r=0$ as a starting point.

\begin{definition}\label{definition:prepared}
A correct process $P_i$ is said to be \texttt{Prepared $(i, s)$} for $m$ on round $r$ when it has delivered all VCBC $(i, s')$ instances tagged with a priority $s' \leq s$, where $s$ corresponds to the lowest priority assigned by $P_i$ to a message $m \notin S_r^{(i)}$. The value of $m$ corresponds to the result of VCBC $(i, s)$.
\end{definition}

\noindent
Note that stating that a process is \texttt{Prepared$(i, s)$} for $m$, is the equivalent of invoking the \texttt{Peek} operation over the priority queue $Q_i$ and obtaining the value $m \neq \bot$, located in its head slot of index $s$.

\subsection{Agreement}
The proof of agreement is structured as Lemma~\ref{lemma:prep-consistency} and Theorem~\ref{theorem:agreement}.
Lemma~\ref{lemma:prep-consistency} is used to prove agreement on the value of \texttt{Prepared$(i, -)$} for any agreement round where consensus holds at entry and Theorem~\ref{theorem:agreement} uses the previous lemma as support to conclude the reasoning.

\begin{lemma}\label{lemma:prep-consistency}
If some correct process has \texttt{Prepared$(i, -)$} for value $m$ on round $r$, then no other correct process can have \texttt{Prepared$(i, -)$} for $m'$, on the same agreement round, where consensus holds at entry, such that $m \neq m'$.
\end{lemma}

\noindent
For this proof we assume that two correct processes $P$ and $P'$ are \texttt{Prepared$(i, -)$} for values $m$ and $m'$ respectively, on round $r$ where consensus holds at entry, and proceed to demonstrate that $m \neq m'$ leads to a contradiction due to a violation the consistency property of VCBC.
\vspace{1ex}

\begin{proof}\label{proof:prep-consistency}
    \beforePfSpace{1ex, 1ex}
    \afterPfSpace{1ex, 2ex}
    \interStepSpace{2ex}
    \pflongnumbers
    
    \step{1}{
        \assume{
            \begin{pfenum}
                \item Both $P$ and $P'$ are correct processes. \label{proof:prep-consistency:assump:1}
                \item Process $P$ is \texttt{Prepared$(i, s)$} for $m$ on round $r$. \label{proof:prep-consistency:assump:2}
                \item Process $P'$ is \texttt{Prepared$(i, s')$} for $m'$ on round $r$. \label{proof:prep-consistency:assump:3}
                \item Consensus holds for $r$. \label{proof:prep-consistency:assump:4}
            \end{pfenum}
        }
        \prove{$m = m'$.}
    }
    
    \step{2}{
        Process $P$ delivered $m$ for VCBC $(i, s)$.
    }
    \begin{proof}
		\pf\
		This follows directly from assumptions \ref{proof:prep-consistency:assump:1} and \ref{proof:prep-consistency:assump:2}, and the definition of the \texttt{Prepared} predicate presented in Definition~\ref{definition:prepared}.
	\end{proof}
	
	\step{3}{
	    Process $P'$ delivered $m'$ for VCBC $(i, s')$.
	}
	\begin{proof}
		\pf\ The same argument of \stepref{2} applies here.
	\end{proof}
    
    \step{4}{
	    $s = s'$.
	}
	\begin{proof}
		\pf\ By assumptions \ref{proof:prep-consistency:assump:2}-\ref{proof:prep-consistency:assump:4} and Definition~\ref{definition:prepared}.
	\end{proof}
	
	\qedstep
	\begin{proof}
		\pf\ Step \stepref{4} implies that, if $m \neq m'$, then the correct processes $P$ and $P'$ must have delivered different values for the same VCBC instance tagged with $(i, s)$. This contradicts the consistency property of VCBC proving the lemma.
	\end{proof}
\end{proof}

\begin{theorem}[Agreement]\label{theorem:agreement}
If a correct process delivers a message $m$, then all correct processes eventually deliver $m$.
\end{theorem}

\noindent
Let us assume, without loss of generality, that some correct process $P$ has delivered a message $m$, during a given round $r$ for which consensus holds at entry. Based on the assumption, it suffices to demonstrate, without loss of generality, that a second correct process $P'$, which is yet to deliver $m$, eventually does so.
\vspace{1ex}

\begin{proof}\label{proof:agreement}
    \beforePfSpace{1ex, 1ex}
    \afterPfSpace{1ex, 2ex}
    \interStepSpace{2ex}
    \pflongnumbers
    
    \step{1}{
        \assume{
            \begin{pfenum}
                \item Both $P$ and $P'$ are correct processes.
                \label{proof:agreement:a1}
                
                \item Process $P$ invoked \texttt{DELIVER$(m)$} for round $r$.
                \label{proof:agreement:a2}
                
                \item Consensus holds for $r$.
                \label{proof:agreement:a3}
            \end{pfenum}
        }
        \prove {Process $P'$ invokes \texttt{DELIVER$(m)$}.}
    }
    
    \step{2}{
        Process $P'$ decides $1$ for ABA $(r)$.
    }
    \begin{proof}
        \step{2.1}{Process $P$ decided $1$ for ABA $(r)$.}
        \begin{proof}
            \pf\ By assumptions \ref{proof:agreement:a1} and \ref{proof:agreement:a2}, that $P$ is correct and follows the protocol rules, therefore only delivering a message for any round $r$ if the corresponding ABA execution decided for $1$.
        \end{proof}
        
        \qedstep
        \begin{proof}
            \pf\ By \stepref{2.1} and the agreement and termination properties of ABA.
        \end{proof}
    \end{proof}
    
    We now have two cases to consider. One where process $P'$ is \texttt{Prepared($F(r), s$)} for some value $m'$ at the beginning of round $r$ and one where it hasn't prepared yet. Note that $F(r)$ corresponds to the results of the queue mapping function for round $r$.
    \vspace{1ex}
    
    \step{3}{
        \case{Process $P'$ is \texttt{Prepared($F(r), s$)} for $m'$.}
        \begin{proof}
            \step{3.1}{
                $m = m'$.
            }
            \begin{proof}
                \pf\ By assumption \stepref{3} and Lemma~\ref{lemma:prep-consistency}.
            \end{proof}
            
            \qedstep
            \begin{proof}
                \pf\ By \stepref{3.1}, and the assumption that $P'$ is correct, process $P'$ delivers $m$ for round $r$.
            \end{proof}
        \end{proof}
    }
    
    \step{4}{
        \case{Process $P'$ is not \texttt{Prepared($F(r), s$)}.}
        \begin{proof}
            \step{4.0}{
                Process $P'$ broadcasts a $\langle \texttt{FILL-GAP}, F(r), - \rangle$ request.
            }
            \begin{proof}
                \pf\ By \stepref{2} and \stepref{4}. Any correct process that hasn't \texttt{Prepared($F(r), s$)} upon a ABA $(r)$ decision for $1$ initiates a recovery sub routine, which starts with the broadcast of a \texttt{FILL-GAP} request for the queue $F(r)$.
            \end{proof}
            
            \step{4.1}{
                Process $P'$ receives a valid $\langle \texttt{FILLER}, - \rangle$ reply.
            }
            \begin{proof}
                \step{4.1.1}{
                    At least one correct process $P_c$ input $1$ to ABA $(r)$.
                }
                \begin{proof}
                    \pf\ By \stepref{2} and the validity property of ABA.
                \end{proof}
                
                \step{4.1.2}{
                    Process $P_c$ was \texttt{Prepared($F(r), s$)} at the start of round $r$.
                }
                \begin{proof}
                    \pf\ By \stepref{4.1.1} and the assumption that $P_c$ is correct.
                \end{proof}
                
                \qedstep
                \begin{proof}
                    \pf\
                    By \stepref{4.1.2} and \cref{definition:prepared}, the correct process $P_c$ must have VCBC delivered all proposals from $P_{F(r)}$ tagged with a priority value $s' \leq s$.
                    Therefore, by the verifiability of VCBC, process $P_c$, can produce a \texttt{FILLER} message that completes all VCBC instances whose result $P'$ is unaware.
                \end{proof}
            \end{proof}
            
            \step{4.2}{
                Process $P'$ becomes \texttt{Prepared$(F(r), s)$} for $m'$.
            }
            \begin{proof}
                \pf\ By \stepref{4.1} and the verifiability property of VCBC.
            \end{proof}
            
            \qedstep
            \begin{proof}
                \pf\ By \stepref{4.2}, the same argument of \stepref{3} applies here.
            \end{proof}

        \end{proof}
    }
\end{proof}

\subsection{Integrity}
\begin{theorem}[Integrity]\label{theorem:integrity}
Every correct process delivers any message $m$ at most once.
\end{theorem}

\noindent
For this proof we assume, without loss of generality, that a correct process $P_i$ has delivered a message $m$ for a given agreement round $r$, and proceed to demonstrate that it cannot deliver $m$ for any subsequent round due to its inability to become \texttt{Prepared$(-, -)$} for $m$, which is a prerequisite for delivery according to the protocol spec.

\begin{proof}\label{proof:integrity}
    \beforePfSpace{1ex, 1ex}
    \afterPfSpace{1ex, 2ex}
    \interStepSpace{2ex}
    \pflongnumbers

    \step{1}{
        \assume{
            \begin{pfenum}
                \item Process $P_i$ is correct.
                \item Process $P_i$ invoked \texttt{DELIVER$(m)$} for round $r$.
            \end{pfenum}
        }
        \prove {Process $P_i$ cannot deliver $m$ for any round $r' > r$.}
    }
    
    \step{2}{
        $S_{r+1}^{(i)} = S_{r}^{(i)} + \{m\}$.
    }
    \begin{proof}
        \pf\ By assumptions 1 and 2, process $P_i$ is correct and has delivered $m$ for round $r$. Therefore, updating its delivered set for the next round to include $m$.
    \end{proof}
    
    \step{3}{
        Process $P_i$ cannot have \texttt{Prepared$(F(r'), -)$} for value $m$, for any round $r' > r$.
    }
    \begin{proof}
        \pf\ By \stepref{2} and Definition~\ref{definition:prepared}, it is impossible for a correct process $P_i$ to prepare for any value $m$ during $r'$ if $m \in S_{r'}^{(i)}$.
        
    \end{proof}
    
    \qedstep
    \begin{proof}
        \pf\
        To deliver $m$ for a round $r'$, the correct process $P_i$ must have \texttt{Prepare$(F(r'), -)$} for the value $m$ during $r'$.
        By \stepref{3}, this is impossible if $P_i$ has delivered $m$ for a prior round since $m \in S_{r'}^{(i)}$, thereby proving the theorem.
    \end{proof}
\end{proof}

\subsection{Validity}
The proof of validity is structured as Lemma~\ref{lemma:prepared-term} and Theorem~\ref{theorem:validity}.
Lemma~\ref{lemma:prepared-term} formalizes an upper bound on the number of agreement rounds required to deliver a message when certain preconditions are met. Theorem~\ref{theorem:validity} builds upon this lemma to conclude the proof.

\begin{lemma}\label{lemma:prepared-term}
If $2f+1$ correct processes have \texttt{Prepared$(i, s)$} for value $m$ by the time they enter round $r$, then $m$ is guaranteed to be delivered by the next round $r' \geq r$, such that $F(r') = i$.
\end{lemma}

\begin{proof}
    \beforePfSpace{1ex, 1ex}
    \afterPfSpace{1ex, 2ex}
    \interStepSpace{2ex}
    \pflongnumbers
    
    \step{1}{
        \assume{
            \begin{pfenum}
                \item Process $P$ is correct.
                \item At least $2f+1$ correct processes are \texttt{Prepared$(i, s)$} for $m$ by round $r$.
                \item $r' \geq r : F(r') = i$.
                \item Consensus holds for $r$ and $r'$.
            \end{pfenum}
        }
        \prove {Process $P$ must have delivered $m$ by the end of round $r'$.}
    }
    
    \step{2}{
        At least $2f+1$ processes input 1 into ABA $(r')$.
    }
    \begin{proof}
        \pf\ This follows from assumption 2, that $2f+1$ processes have \texttt{Prepared$(i, s)$} on entry to round $r$, therefore proposing $1$ into the next ABA execution pertaining to queue $i$, which by assumption 3 corresponds to $r'$.
    \end{proof}
    
    \step{3}{
        Process $P$ decides $1$ for ABA $(r')$.
    }
    \begin{proof}
        \pf\ By \stepref{2}, and the validity and termination properties of ABA.
    \end{proof}
    
    \qedstep
    \begin{proof}
        \pf\ The same reasoning from the steps 3 and 4 from \Cref{theorem:agreement} can be applied here.
    \end{proof}
\end{proof}

\begin{theorem}[Validity]\label{theorem:validity}
If a correct process broadcasts a message $m$, then some correct process eventually delivers $m$.
\end{theorem}

\begin{proof}
    \beforePfSpace{1ex, 1ex}
    \afterPfSpace{1ex, 2ex}
    \interStepSpace{2ex}
    \pflongnumbers
    
    \step{1}{
        \assume{
            \begin{pfenum}
                \item Both $P$ and $P_i$ are correct processes.
                \item Process $P_i$ invoked \texttt{SEND$(m)$}.
            \end{pfenum}
        }
        \prove {Process $P$ invokes \texttt{DELIVER$(m)$}}
    }
    
    \step{2}{
        Every correct process delivers $m$ for VCBC $(i, s)$ .
    }
    \begin{proof}
        \pf\
        Let $s$ denote the priority value attributed by $P_i$ to $m$.
        Then, by the assumption that $P_i$ is correct and the validity property of VCBC all correct processes deliver $m$ tagged with $(i, s)$.
    \end{proof}
    
    \step{3}{
        All correct processes become \texttt{Prepared$(i, s)$} for $m$.
    }
    \begin{proof}
        \pf\ This follows directly from \stepref{2}.
    \end{proof}
    
    \qedstep
    \begin{proof}
        \pf\
        By \stepref{3}, which fulfills the pre-requirements for \Cref{lemma:prepared-term} to hold.
    \end{proof}
\end{proof}

\subsection{Total Order}
\begin{theorem}[Total Order]\label{theorem:total-order}
If two correct processes deliver messages $m$ and $m'$ then both processes deliver $m$ and $m'$ in the same order.
\end{theorem}

\noindent
For this proof, we assume without loss of generality that two correct processes $P$ and $P'$ have respectively delivered $m$ and $m'$, for round $r$ where consensus holds on entry. Based on this assumption it suffices to prove that $m = m'$, in order to demonstrate total order.

\begin{proof}
    \beforePfSpace{1ex, 1ex}
    \afterPfSpace{1ex, 2ex}
    \interStepSpace{2ex}
    \pflongnumbers
    
    \step{1}{
        \assume{
            \begin{pfenum}
                \item Both $P$ and $P'$ are correct processes.
                \item Process $P$ invoked \texttt{DELIVER$(m)$} for round $r$.
                \item Process $P$ invoked \texttt{DELIVER$(m')$} for round $r$.
                \item Consensus holds for $r$.
                \item $F(r) = i$.
            \end{pfenum}
        }
        \prove {$m = m'$.}
    }
    
    \step{2}{
        Process $P$ has \texttt{Prepared$(i, s)$} for $m$ during round $r$.
    }
    \begin{proof}
        \pf\ This follows from the assumption the $P$ is correct and has invoked \texttt{DELIVER$(m)$} for round $r$, since preparing for a value always precedes its delivery.
    \end{proof}
    
    \step{3}{
        Process $P'$ has \texttt{Prepared$(i, s)$} for $m'$ during round $r$.
    }
    \begin{proof}
        \pf\ The same argument as \stepref{2} applies here.
    \end{proof}
    
    \qedstep
    \begin{proof}
        \pf\ By \stepref{2} and \stepref{3}, and Lemma~\ref{lemma:prep-consistency}.
    \end{proof}
\end{proof}
\end{appendices}

\end{document}